\documentclass[a4paper,reqno]{article}

\usepackage{tgschola} 
\usepackage{latexsym}
\usepackage[english]{babel}
\usepackage{fancyhdr}
\usepackage[mathscr]{eucal}
\usepackage{amsmath}
\usepackage{mathrsfs}
\usepackage{amsthm}
\usepackage[left=2.5cm, right=2.5cm, top=2.5cm, bottom=2cm]{geometry}
\usepackage{amsfonts}
\usepackage{amssymb}
\usepackage{amscd}
\usepackage{bbm}
\usepackage{graphicx}
\usepackage{graphics}
\usepackage{latexsym}
\usepackage{color}
\usepackage{calc}
\usepackage{cancel}
\usepackage{marvosym}
\usepackage{tikz}           
\usetikzlibrary{arrows.meta, positioning}
\usepackage[normalem]{ulem}
\usepackage{bbold}

\newcommand{\ud}{\mathrm{d}}

\newcommand{\ii}{\mathrm{i}}

\definecolor{azzurro}{rgb}{0.19, 0.55, 0.91}
\definecolor{giallino}{rgb}{0.97,	0.78,	0.05}
\definecolor{verdino}{rgb}{0.49, 0.74, 0.54}

\definecolor{verdissimo}{rgb}{0.0, 0.5, 0.0}

\theoremstyle{plain}
\newtheorem{theorem}{Theorem}[section]
\newtheorem{lemma}[theorem]{Lemma}

\usepackage{tocloft}

\renewenvironment{thebibliography}[1]{
  \begin{oldthebibliography}{#1}
    \setlength{\itemsep}{0.1em}
    \setlength{\parskip}{0em}
}
{
  \end{oldthebibliography}
}

\theoremstyle{definition}

\newtheorem{remark}[theorem]{Remark}
\newtheorem*{remark*}{Remark}

\newcommand{\OOO}{\mathcal{O}}

\numberwithin{equation}{section}

\addtocontents{toc}{\protect\setcounter{tocdepth}{2}}

\begin{document}\hyphenation{di-men-sion-al}

\title{\textbf{On Exponentially Long Prethermalization Timescales in Isolated Quantum Systems}}

\author{Matteo Gallone\footnote{\emph{Email:} \texttt{matteo.gallone@unimi.it}} 
\\ \vspace{-10pt} \\
\small \textit{\textsuperscript{*}Dipartimento di Matematica ``F.~Enriques'', Universit\`a degli Studi di Milano, via Saldini 50 -- 20133 Milano (Italy) }\\ \vspace{-18pt} \\ \small %
}

\date{\today}

\maketitle

\vspace{-10pt}

\begin{abstract}
We study prethermalization in time-independent quantum many-body systems on a $d$-dimensional lattice with an extensive local Hamiltonian $H=N+\varepsilon P$, in the regime where $\varepsilon \ll 1$. {We prove that the prethermal timescale is exponential in $\varepsilon_0/\varepsilon$, where $\varepsilon_0$ is an explicit scale determined by the locality exponent and the local norm of $P$. We also construct two extensive local quasi-conserved quantities whose variation remains exponentially small throughout this time
interval.}
\end{abstract}

\tableofcontents

\section{Introduction and main results}
The approach to thermal equilibrium is a central problem in statistical mechanics. Originally formulated by Boltzmann~\cite{boltzmann1872}, it has attracted renewed attention following the {unexpected outcome} of the numerical experiments by Fermi, Pasta, Ulam, and Tsingou~\cite{FPU-Original}, as well as the experimental realization of the ``quantum Newton's cradle'' by Kinoshita, Wenger, and Weiss~\cite{Kinoshita2006-wd}.

It is by now well understood that the non-equilibrium dynamics of quantum many-body systems can exhibit rich and complex behavior. In particular, a wide class of systems displays \emph{prethermalization}: after a rapid initial evolution, the system relaxes to a long-lived quasi-stationary state, referred to as a \emph{prethermal state}, before eventually reaching thermal equilibrium or, more generally, its asymptotic steady regime on much longer timescales.

A class of systems exhibiting prethermalization consists of quantum many-body systems subjected to an external periodic driving \cite{Holthaus2015,Bukov2015,Eisert2015,Kuwahara2016,
Mori2016,Moessner2017,Mallayya2019,Oka2019,Rudner2020,Harper2020,Collura2022,
Ho2023,Qi2026}. In general, such systems lack nontrivial conserved quantities, including energy, and are therefore expected to evolve towards a featureless infinite-temperature Gibbs state. However, it has been shown theoretically~\cite{Eckardt2005,DAlessio2014,Else2016,
Potirniche2017,Yao2017,Weidinger2017,Machado2019,Choi2020,Machado2020,Ye2021,Zhuo-preprint}, experimentally observed~\cite{
Rechtsman2013,Zhang2017,Choi2017,
Singh2019,RubioAbadal2020,Geier2021,Xiao-Mi-2021,Kyprianidis2021,Randall2021,Peng2021,
Beatrez2021} and established mathematically~\cite{Abanin2017,Ho2018,Else2020}
that when the driving frequency $\omega$ is large compared to the local energy scale $J$, energy absorption is strongly suppressed. In this regime, the system relaxes to a quasi-stationary prethermal state whose lifetime is exponentially long in $|\omega|/J$. 

A central feature of prethermalization is that, within the prethermal regime, the system may display nontrivial macroscopic properties and support genuine nonequilibrium phases of matter. This perspective motivates the study of \emph{Floquet systems} and the program of \emph{Floquet engineering}, in which suitably designed periodic drivings give rise to effective Hamiltonians with controlled properties \cite{Potter2016,Ye2021,Eckhardt2022,Zhang2022,Zhang2022e}. 

The possibility of engineering such phases has further motivated the study of systems subjected to quasi-periodic drivings~\cite{Blekher1992,Dumitrescu2018,Else2020,Qi2021,Long2021,Zhao2021,Mori2021,He2023,Gallone-Langella-2024,Kumar2024,He2025,Fang2025,Dutta2025,Qp-experiment-2025,QP-esperimento2,AltroQP}. In this setting, when the driving frequencies are non-resonant and large compared to the local energy scale, both theoretical and experimental works indicate that prethermalization persists over long timescales that depend sensitively on the regularity of the driving protocol. More precisely, for analytic drivings the prethermal timescale is typically stretched exponential in $|\omega|/J$, whereas for drivings with finite differentiability it has been observed experimentally \cite{He2023,He2025} and explained theoretically \cite{Mori2021,Gallone-Langella-Ck} that the relevant timescales grow only as a power law.

Prethermalization in isolated systems also plays a fundamental role in the foundations of statistical mechanics, as it provides a mechanism for delayed thermalization, reminiscent of the behavior observed in the original Fermi--Pasta--Ulam--Tsingou problem. A common framework in which prethermalization arises is that of nearly integrable (or quasi-integrable) quantum systems. In such systems, the dynamics typically leads, on intermediate timescales, to a prethermal state that can be described by a suitable Generalized Gibbs Ensemble (GGE), reflecting the presence of (approximate) conserved quantities \cite{Kollar2011,Marcuzzi2013,
Bertini2015,Mori2018,Bertini2020,Birnkammer2022,
Surace2023,Yin2025}. {In this setting as well, the prethermal state can exhibit phenomena such
as} integrable turbulence \cite{Gallone2022PRL,Gallone-Random} or asymptotic localization \cite{DeRoeck2024,DeRoeck2025}. Prototypical examples include systems with spectral gaps or weak perturbations of integrable models.

{A complementary form of prethermal stability for autonomous gapped systems has been developed in the theory of non-equilibrium almost-stationary states and of locally robust gapped sectors for certain classes of initial states \cite{Teufel2019,Yin2023}.}

In this work, we consider quantum many-body systems whose Hamiltonian can be written in the form
\begin{equation}\label{eq:E-414}
    H = N + \varepsilon P,
\end{equation}
where $N$ and $P$ are extensive local self-adjoint operators and $N$ has integer spectrum. For this class of models, it was shown in~\cite{Abanin2017} that prethermalization occurs up to times {$|t| \lesssim \exp\!\big(-\frac{1}{\varepsilon\,\ln^3 \varepsilon}\big)$}, during which the dynamics admits two approximate constants of motion. 

The main result of the present paper is an improvement of this bound: we prove that the prethermal timescale is in fact exponential in $\varepsilon^{-1}$, and that the corresponding approximate constants of motion can be chosen to be extensive local operators. {The proof is closely inspired by resonant normal-form techniques for classical Hamiltonian systems \cite{Gallone2022} and is reminiscent of the ``isochronous'' version of Nekhoroshev’s
theorem \cite{Nek79,Benettin-Gallavotti}; see \cite{Bambusi-Langella-Cinf} for a streamlined proof.}

{
\paragraph{Notation.}
All the notation used in the manuscript is standard. We write $\langle x \rangle:=(1+|x|)$.}

\subsection{Setting and main results}
We consider quantum many-body systems on a finite $d$-dimensional lattice $\Lambda=\mathbb{Z}^d \cap [-L,L]^d$. The Hilbert space of the system is $\mathcal{H}_\Lambda:=\bigotimes_{x \in \Lambda} \mathfrak{h}_x$ where $\mathfrak{h}_x$ is the on-site Hilbert space. For simplicity, we assume that {all on-site Hilbert spaces are identical and finite-dimensional}, i.e.~$\mathfrak{h}_x=\mathfrak{h} \cong \mathbb{C}^q$  for some $q \in \mathbb{N}$. {We will be particularly interested in \emph{extensive local operators}, that are operators $A$ supplied with a fixed decomposition $A=\sum_{S \in \mathcal{P}_c(\Lambda)}$ where $\mathcal{P}_c(\Lambda)$ denotes the family of connected subsets of $\Lambda$, and $A_S$ acts nontrivially only on $S$. Throughout the paper, when such an operator is called self-adjoint, we mean that its prescribed local terms are self-adjoint: $A_S = A^*_S$ for every $S$. 

A standard way to quantify the locality of an extensive local operator is to define a norm which relies on the collection $A_S$ of its local parts as
\begin{equation}\label{eq:ZioPinoloSecondo}
	\Vert A \Vert_{\kappa}:= \sup_{x \in \Lambda} \sum_{\substack{S \in \mathcal{P}_c(\Lambda) \\ x \in S}} \Vert A_S \Vert_{\mathrm{op}} e^{\kappa |S|} \, .
\end{equation}
We say that $A \in \OOO_{\kappa}$, if there exists a constant $C>0$ \emph{independent of $|\Lambda|$} such that $\Vert A \Vert_{\kappa} < C$. If $A=A_S$ for some $S \in \mathcal{P}_c(\Lambda)$, we say that $A$ is a local observable \emph{supported} on $S$. Throughout the paper, the local decomposition is fixed and regarded as part of the data.

We study systems whose Hamiltonian is a self-adjoint extensive local operator of the form}
\begin{equation}\label{eq:OriginalHamiltonian}
	H=N+\varepsilon P
\end{equation}
where $P$ is a {self-adjoint} extensive local operator, that is for some $\kappa>0$ we have $P \in \OOO_{\kappa}$, and $N=\sum_{x \in \Lambda} N_x$ is a ``number operator'', that is $\sigma(N_x) \subset{\mathbb{Z}}$ {for all $x \in \Lambda$}. 

We introduce the positive parameter
\begin{equation}\label{eq:EpsZero}
	\varepsilon_0 {=} \frac{(e-1){\kappa}}{64 \pi (3(e-1)+e^2) \Vert P \Vert_{\kappa} } \, ,
\end{equation}
which sets the scale of the effective spectral gap relative to the perturbation strength.

We now state our main result, which establishes exponentially long prethermalization and the existence of quasi-conserved quantities.
\begin{theorem}\label{thm:Main}
	Let $\kappa>0$ and {let} $H{=N+\varepsilon P}$ be {as} in \eqref{eq:OriginalHamiltonian}, {where $N=\sum_{x \in \Lambda} N_x$. Assume that each $N_x$ is supported on $\{x\}$, that $N,P \in \OOO_\kappa$ are self-adjoint and that $\sigma(N_x) \subset \mathbb{Z}$ for any $x \in \Lambda$}. Assume that ${0<}\varepsilon < \varepsilon_0$. Then there exist two extensive local {self-adjoint} operators $\mathcal{N},\mathcal{Z} \in \OOO_{\frac{\kappa}{2}}$ and four positive constants $\mathsf{C}_1$, $\mathsf{C}_2$, $\mathsf{C}_3$ and $\mathsf{C}_4$ {independent of $|\Lambda|$} such that
	\begin{itemize}
		\item[(i)] \emph{Structure:} $\mathcal{N}$ has integer spectrum $\sigma(\mathcal{N}) \subset \mathbb{Z}$;
		\item[(ii)] \emph{Closeness to $N$:} $\mathcal{N}$ is close to $N$ in operator norm, i.e. $|\Lambda|^{-1}\Vert \mathcal{N}-N \Vert_{\mathrm{op}} \leq \mathsf{C}_1 \varepsilon \Vert N \Vert_{\kappa}$. 
		\item[(iii)] \emph{Quasi conservation:} {The variations of $\mathcal{N}$ and $\mathcal{Z}$ remain exponentially small for exponentially long times:}
		\begin{equation}\label{eq:MainBound333}
			\frac{1}{|\Lambda|} \Vert e^{\ii H t} \mathcal{N} e^{-\ii H t} - \mathcal{N} \Vert_{\mathrm{op}} \leq \mathsf{C}_2 \Vert N \Vert_{\frac{3 \kappa}{4}} \, |t| \, \varepsilon e^{-\frac{\varepsilon_0}{\varepsilon}} \, , \quad \frac{1}{|\Lambda|} \Vert e^{\ii H t} \mathcal{Z} e^{-\ii H t} - \mathcal{Z} \Vert_{\mathrm{op}} \leq \mathsf{C}_3 |t| \, \varepsilon^2 \, e^{-\frac{\varepsilon_0}{\varepsilon}} \, .
		\end{equation}
	\item[(iv)] \emph{Commutativity:} $[\mathcal{N},\mathcal{Z}]=0$.
	\item[(v)] \emph{Effective dynamics:} Let $H_{\mathrm{eff}}=\mathcal{N}+\mathcal{Z}$. Then, the dynamics generated by $H_{\mathrm{eff}}$ approximates the true dynamics on local observables: let $O$ be a local observable supported on $S$, then
	\begin{equation}
		\Vert e^{\ii H t}Oe^{-\ii H t}-e^{\ii (\mathcal{N}+\mathcal{Z})t} O e^{-\ii (\mathcal{N}+\mathcal{Z})t} \Vert_{\mathrm{op}} \leq \mathsf{C}_4 {\Vert O \Vert_{\mathrm{op}}} \varepsilon \, , \qquad \forall |t| \leq e^{\frac{\varepsilon_0}{(d+1) \varepsilon}} \, .
	\end{equation}
	{The constant $\mathsf{C}_4$ depends on $|S|$.}
	\end{itemize}
\end{theorem}

\begin{remark}\textbf{}\label{remarkio}
	\begin{itemize}
			\item[(i)] All constants in Theorem \ref{thm:Main} are independent of the size $|\Lambda|$ which means that the results hold in the thermodynamic limit. {Indeed, since all constants and the local norms are uniform in $|\Lambda|$, one can take the thermodynamic limit on all the estimates of Theorem \ref{thm:Main}. We also underline that Theorem \ref{thm:Main} does not provide any information on the dynamics in the infinite volume case, in particular does not provide information on the limit $L\to+\infty$ of the operators. Constructing an infinite-volume dynamics and passing these estimates to that limit require compatible families of interactions and a separate limiting argument; these issues are not addressed here. The constants in Theorem \ref{thm:Main} can be chosen as follows:}
		\[
			\begin{split}
				\mathsf{C}_1&=256 \pi  \frac{e}{e-1}  \frac{e^{-\frac{\kappa}{4}}}{\kappa} \Vert P \Vert_{\kappa} \, , \\
				\mathsf{C}_2&=\frac{16}{3 \kappa} \Vert P \Vert_{\kappa} \, , \\
				\mathsf{C}_3&= \frac{16}{3 \kappa} \frac{e}{e-1} \Vert P \Vert_{\kappa}^2 \, ,\\
				\mathsf{C}_4&=8 \pi \frac{e}{e-1} \Vert P \Vert_\kappa |S|+2^{d+1} C_{LR} \langle \Vert N \Vert_\kappa+\varepsilon_0 {\textstyle\frac{e}{e-1}} \Vert P \Vert_{\kappa} \rangle^d e \Vert P \Vert_\kappa \, ,
			\end{split}		
		\]
		where $C_{LR}$ is the constant of Lemma \ref{lem:LibRobinson}. {We also remark that $\mathsf{C}_1$, $\mathsf{C}_2$ and $\mathsf{C}_3$ are independent of $\Vert N \Vert_{\kappa}$.}
		\item[(ii)] {Items (ii) and (iii) also imply the quasi-conservation of $N$  for} exponentially long times, namely
		\begin{equation}
			\frac{1}{|\Lambda|} \Vert e^{\ii H t} N e^{-\ii H t}-N \Vert_{{\mathrm{op}}} \leq {(2 \mathsf{C}_1+\mathsf{C}_2) \Vert N \Vert_{\kappa} \varepsilon} \, , \qquad \forall |t| \leq e^{\frac{\varepsilon_0}{\varepsilon}} \, .
		\end{equation}
		{Nevertheless, such a bound can be improved by using the conservation of energy only. Indeed, using that $N=H-\varepsilon P$, one has
		\begin{equation}
			\frac{1}{|\Lambda|} \Vert e^{\ii H t} N e^{-\ii H t}-N\Vert_{\mathrm{op}} = \frac{\varepsilon}{|\Lambda|} \Vert -e^{\ii H t}Pe^{-\ii H t}+P \Vert_{\mathrm{op}} \leq 2 \varepsilon \Vert P \Vert_0 \, , \qquad \forall t \in \mathbb{R}.
		\end{equation}
		Thus, the nontrivial content in Theorem \ref{thm:Main} is the quasi-conservation of the \emph{dressed} quantity $\mathcal{N}$ and the existence of another quasi-conserved quantity $\mathcal{Z}$.}
		\item[(iii)] The result is based on a normal form transformation $H \mapsto YHY^*$ where $Y=e^{-\ii G^{(n_*-1)}} \cdots e^{-\ii G^{(0)}}$ is {is a composition of time-one unitary flows generated by self-adjoint extensive local operators} $G^{(j)}$. We prove in Lemma \ref{lem:Iterative} that $YHY^*=N+Z^{(n_*)}+P^{(n_*)}$ where $P^{(n_*)} = \mathscr{O}(e^{-\frac{\varepsilon_0}{\varepsilon}})${,} $Z^{(n_*)}=\mathscr{O}(\varepsilon)$ and $[Z^{(n_*)},N]=0$. The dressed operators are defined as $\mathcal{N}=Y^* N Y$ and $\mathcal{Z}=Y^* Z^{(n_*)} Y$.
		\item[(iv)] When reading the second equation in Item \emph{(iii)}, one should keep in mind that $\mathcal{Z} =\mathscr{O}(\varepsilon)$ and therefore the quasi-conservation of $\mathcal{N}$ and $\mathcal{Z}$ holds on the same time-scale. 
		\item[(v)] Our Theorem \ref{thm:Main} {improves Theorem 3 in \cite{Abanin2017} in two respects. First, Theorem 3 of \cite{Abanin2017} yields timescale $|t| \leq e^{\frac{\varepsilon_0}{\varepsilon{(1-\ln( \varepsilon/\varepsilon_0))^3}}}$, whereas our estimate removes the logarithmic correction. Second, our resulting dressed operators belong to a fixed class $\OOO_{\frac{\kappa}{2}}$ with a locality exponent that does not tend to zero as $\varepsilon \to 0$.}
		\item[(vi)] {The estimate for $\mathcal{Z}$ in Item (iii) does not depend explicitly on $\Vert N \Vert_{\kappa}$. Subject to suitable domain assumptions for the unbounded operators and their commutators, the same strategy can be extended to separable infinite-dimensional on-site Hilbert spaces.}
		\item[(vii)] With simple modifications of the proof, Theorem \ref{thm:Main} can be extended to the case of number operators $N$ which are sum of local {terms} supported on larger sets $S$. In this {setting, such a result can be proved assuming that the perturbation is} \emph{strongly local} (see e.g.~\cite{DeRoeck-Verreet,Gallone-Langella-2024}). 
	\end{itemize}
\end{remark}

\subsection{An application: Quantum Ising model in strong magnetic field}
We consider the quantum Ising model on the $d$-dimensional lattice $\Lambda=\mathbb{Z}^d \cap [-L,L]^d$. Using the standard notation for Pauli matrices
\begin{equation}
	X=\sigma^{(1)}=\begin{pmatrix}
		0 & 1 \\
		1 & 0 
	\end{pmatrix} \, , \quad Y=\sigma^{(2)}=\begin{pmatrix}
		0 & -\ii \\
		\ii & 0
	\end{pmatrix} \, , \quad Z=\begin{pmatrix}
		1 & 0 \\
		0 & -1
	\end{pmatrix} \, ,
\end{equation}
the Hamiltonian of the model we consider is
\begin{equation}
	H= \sum_{x \in \Lambda} Z_x -\varepsilon \sum_{\substack{x,y \in \Lambda \\ {\text{s.t. }} |x-y|=1}} X_x X_{y} \, .
\end{equation}
This Hamiltonian {has the on-site form required by Theorem \ref{thm:Main} and} the structure of \eqref{eq:E-414} with 
\begin{equation}
	N=\sum_{x \in \Lambda} Z_x \, , \qquad P={-}\sum_{\substack{x{,y} \in \Lambda \\ {\text{s.t. }|x-y|=1}}} X_x X_y \, .
\end{equation}
{Here $N$ represents the total magnetization in the $z$ direction. Decomposing $P$ into nearest-neighbour bond terms gives $\Vert P \Vert_{\kappa} \leq 2d e^{\kappa}$. Hence, Items (ii) and (iii) of Theorem \ref{thm:Main} imply that, if $\varepsilon < \varepsilon_0=\frac{(e-1){\kappa}}{64 \pi (3(e-1)+e^2) {2d \, }e^{2\kappa} }$, the dressed magnetization $\mathcal{N}$ is quasi-conserved for exponentially long times. More precisely,

\begin{equation}
	\frac{1}{|\Lambda|} \Vert e^{\ii H t} \mathcal{N} e^{-\ii H t} - \mathcal{N} \Vert_{\mathrm{op}} \leq \mathsf{C}_2 \Vert N \Vert_{\frac{3 \kappa}{4}} \varepsilon e^{-\frac{\varepsilon_0}{2 \varepsilon}} \, , \qquad |t| \leq e^{\frac{\varepsilon_0}{2 \varepsilon}} \, .
\end{equation} 

The normal form construction also identifies the leading-order contribution to $\mathcal{Z}$. Indeed, \eqref{eq:BSHomo} gives
\begin{equation}
	\frac{1}{2 \pi} \int_0^{2\pi} e^{-\ii \theta N} X_x X_y e^{\ii N \vartheta} \, d \vartheta=\frac{1}{2}(X_xX_y+Y_xY_y) \, ,
\end{equation}
and therefore
\begin{equation}
	\mathcal{Z}=-\frac{\varepsilon}{2}\sum_{\substack{x,y \in \Lambda \\ \text{s.t. } |x-y|=1}} (X_xX_y+Y_xY_y)+\mathscr{O}(\varepsilon^2) \, .
\end{equation}
Here the reminder $\mathscr{O}(\varepsilon^2)$ is intended as $\Vert \mathcal{Z}-P \Vert_{\frac{\kappa}{2}} \lesssim \varepsilon^2 \Vert P \Vert_\kappa$, which is uniform in the volume. 

As stated in Remark \ref{remarkio}(iii)--(iv), $\mathcal{Z}=\mathscr{O}(\varepsilon)$. Consequently, the two estimates in \eqref{eq:MainBound333} are effective on the same timescale.}

\section{Normal form}

The proof of Theorem~\ref{thm:Main} is based on an iterative (non-convergent) normal form scheme, whose goal is to progressively eliminate the non-resonant part of the Hamiltonian.

Starting from a Hamiltonian $H^{(n)}$, we construct a finite sequence of unitarily equivalent Hamiltonians 
\begin{equation}
	H^{(n+1)} = e^{-\ii G^{(n)}} H^{(n)} e^{\ii G^{(n)}},
\end{equation}
where $G^{(n)}$ is an extensive local self-adjoint operator. At each step, the Hamiltonian is decomposed as
\begin{equation}
	H^{(n)} = N + Z^{(n)} + P^{(n)},
\end{equation}
where $[N, Z^{(n)}] = 0$, while $P^{(n)}$ is a perturbative term of size $\mathscr{O}(e^{-n})$.

The generator $G^{(n)}$ and the correction (resonant part) $\tilde{Z}^{(n)}$ are defined by solving the homological equation
\begin{equation}
	-\ii [G^{(n)}, N] + P^{(n)} = \tilde{Z}^{(n)},
\end{equation}
with the constraint $[\tilde{Z}^{(n)}, N] = 0$. With this choice, one obtains
\[
H^{(n+1)} = N + Z^{(n+1)} + P^{(n+1)},
\]
where $Z^{(n+1)} = Z^{(n)} + \tilde{Z}^{(n)}$ and $P^{(n+1)} = \mathscr{O}(e^{-(n+1)})$.

The number of normal form steps increases as $\varepsilon$ decreases. By optimally truncating the iteration, one obtains a cutoff $n_* {=} \left\lfloor \frac{\varepsilon_0}{\varepsilon} \right\rfloor$ 
so that the remainder satisfies $
P^{(n_*)} = \mathscr{O}(e^{-n_*}) = \mathscr{O}\big(e^{-\frac{\varepsilon_0}{\varepsilon}}\big).$

The extensive locality of the dressed operators $\mathcal{N}$ and $\mathcal{Z}$ follows from that of $N$ and $Z^{(n_*)}$, together with the fact that the unitary transformation $Y$ is given by a composition of flows generated by extensive local operators.

Finally, to obtain estimates that are independent of $\Vert N\Vert_\kappa$, we exploit the fact that each $G^{(n)}$ solves the homological equation. This ensures, in particular, that Item~(iii) of Theorem~\ref{thm:Main} remains valid even for unbounded number operators $N$.

\subsection{Conjugation and homological equation}

As discussed above, a central ingredient of our approach is the solution of the homological equation. In this subsection, we provide bounds on conjugations generated by extensive local self-adjoint operators, and we describe the explicit solution of the homological equation.

Let us first introduce the notation
\begin{equation}\label{eq:AdjAB}
	\mathrm{Ad}_A^0 B := B \, , \qquad 
	\mathrm{Ad}_A^j B :={ [A,\mathrm{Ad}_A^{j-1} B]}, \quad j \geq 1,
\end{equation}
for $A,B \in \OOO_\kappa$. {The conjugation of $H$ by $e^{-\ii G}$ can then be expanded as}
\begin{equation}
	e^{-\ii G} H e^{\ii G} = \sum_{j=0}^{\infty} \frac{(-\ii)^j}{j!} \mathrm{Ad}_G^j H \, .
\end{equation}

{Because every operator is supplied with a fixed local decomposition, we
must specify the induced local decomposition of a commutator. Indeed, if $A,B \in \OOO_{\rho}$ for some $\rho>0$, we have
\begin{equation}
	[A,B]=\sum_{S \in \mathcal{P}_c(\Lambda)} ([A,B])_S \, , \qquad  ([A,B])_S:=\sum_{\substack{S_1 \cup S_2=S \\ S_1 \cap S_2 \neq \varnothing}} [A_{S_1}, B_{S_2}] \, .
\end{equation}}

The following Lemma, taken from Eq.~(A.3) in \cite{Gallone-Langella-2024}, gives bounds on iterated commutators; these bounds form the key technical ingredient {for} Lemma~\ref{lem:Commutatori}.

\begin{lemma}\label{lem:D-345}
	Let $\kappa,\delta>0$ and $A,B \in \OOO_{\kappa+\delta}$. Then{, for $j\geq 1$},
	\begin{equation}\label{eq:ALe601}
		\Vert \mathrm{Ad}_A^j B \Vert_{\kappa} \leq \left( \frac{j}{e}\right)^j \frac{(4 e^{-\kappa})^j}{\delta^j} \Vert A \Vert_{\kappa+\delta}^j \Vert B \Vert_{\kappa+\delta} \, .
	\end{equation}
\end{lemma}

{\begin{proof}
The proof follows that of Eq.~(A.3) in \cite{Gallone-Langella-2024} verbatim, so we omit it. That argument uses only the hypotheses stated in Lemma \ref{lem:D-345}.
\end{proof}}

\begin{lemma}\label{lem:Commutatori}
	Let $A,B \in \OOO_{\kappa+\delta}$ and $\eta \in (0,1)$ satisfy
	\begin{equation}
		\frac{4 e^{-\kappa} \Vert A \Vert_{\kappa+\delta}}{\delta} \leq \eta.
	\end{equation}
	Then the conjugation $e^A B e^{-A}$ satisfies the bounds
	\begin{align}
		\label{eq:Conjugation}\Vert e^A B e^{-A} \Vert_{\kappa} &\leq \frac{1}{1-\eta} \, \Vert B \Vert_{\kappa+\delta}, \\ \label{eq:FirstOrderEstimate}
		\Vert e^A B e^{-A} - B \Vert_{\kappa} &\leq C_\eta e^{-\kappa} \frac{1}{\delta} \Vert A \Vert_{\kappa+\delta} \Vert B \Vert_{\kappa+\delta},
	\end{align}
	with $C_\eta = 4/(1-\eta)$.
{If, moreover, $[A,B] \in \mathcal{O}_{\kappa+\delta}$, then}
	\begin{equation}
		\Vert e^A B e^{-A} - B - [A,B] \Vert_{\kappa} \leq C_\eta e^{-\kappa} \frac{1}{\delta} \Vert A \Vert_{\kappa+\delta} \Vert [A,B] \Vert_{\kappa+\delta}, \label{eq:CommEstimate}	
	\end{equation}
{ with the same constant $C_\eta$ as above.}
\end{lemma}
\begin{proof}
	Equation \eqref{eq:FirstOrderEstimate} follows from {Eq.~(4.6) of} Lemma 4.2 in \cite{Gallone-Langella-2024}.
	
	Equation \eqref{eq:Conjugation} follows from
	\begin{equation*}
		\begin{split}
			\Vert e^A B e^{-A} \Vert_{\kappa} & \leq  \sum_{j=0}^\infty \frac{1}{j!} \Vert\mathrm{Ad}_A^j(B) \Vert_{\kappa} \overset{\text{\eqref{eq:ALe601}}}{\leq} \sum_{j=0}^\infty \frac{1}{j!} \left( \frac{j}{e} \right)^j \frac{(4 e^{-\kappa})^j}{\delta^j} \Vert A \Vert_{\kappa+\delta}^j \Vert B \Vert_{\kappa+\delta} \\
			&\leq \left(\sum_{j=0}^\infty \frac{(4e^{-\kappa})^j}{\delta^j}\Vert A \Vert_{\kappa+\delta}^j\right) \Vert B \Vert_{\kappa+\delta} \leq \frac{1}{1-\eta} \Vert B \Vert_{\kappa+\delta} \, ,
		\end{split}
	\end{equation*}
	where we used Stirling formula $j!\geq \left(\frac{j}{e}\right)^j \sqrt{2 \pi j}$ and summed the {resulting} geometric series.
	
	{To prove \eqref{eq:CommEstimate}, we write}
	\[
		\begin{split}
			 e^{A} B e^{-A}-B-[A,B]  &{=} \sum_{j=0}^\infty \frac{1}{j!}\mathrm{Ad}^j_A(B)-B-[A,B] =\sum_{j=2}^\infty \frac{1}{j!}\mathrm{Ad}^j_A(B) \\
			 &=\sum_{j=2}^\infty \frac{1}{j!}\mathrm{Ad}^{j-1}_A([A,B])=\sum_{r=0}^\infty \frac{1}{(r+2)!} \mathrm{Ad}_A^{r+1}([A,B]) \, .
		\end{split}
	\]
	{Taking} the $\Vert \cdot \Vert_{\kappa}$ norm from both sides {gives}
	\[
		\begin{split}
			\Vert e^{A} B e^{-A}-B-[A,B]  \Vert_{\kappa} &\leq \sum_{r=0}^\infty \frac{1}{(r+2)!} \Vert \mathrm{Ad}_A^{r+1}([A,B]) \Vert_{\kappa} \\
			&\hspace{-3pt}\overset{\text{\eqref{eq:ALe601}}}{\leq} \sum_{r=0}^\infty \frac{1}{(r+2)!} \left(\frac{{r+1}}{e} \right)^{r+1} \frac{(4 e^{-\kappa})^{r+1}}{\delta^{r+1}} \Vert A \Vert_{\kappa+\delta}^{r+1} \Vert [A,B] \Vert_{\kappa+\delta} \\
			&\leq \left(\sum_{r=0}^\infty \frac{(4e^{-\kappa})^r}{\delta^r} \Vert A \Vert^{r}_{\kappa+\delta} \right) \frac{4 e^{-\kappa}}{\delta} \Vert A \Vert_{\kappa + \delta} \Vert [A,B] \Vert_{\kappa+\delta}
		\end{split}
	\]
	where we used Stirling formula $r! \geq (\frac{r}{e})^r \sqrt{2 \pi r}$ {and $(\frac{e}{r+2})^{r+2} (\frac{r+1}{e})^{r+1} \frac{1}{\sqrt{2 \pi (r+2)}} \leq \frac{e}{\sqrt{4 \pi}}<1$}. {Summing the geometric series proves  Eq.~\eqref{eq:CommEstimate}.}
\end{proof}

Next, we {study the solutions to} the homological equation, which amounts to finding extensive local operators $G$ and $B$ such that
\begin{equation}\label{eq:homEqGenericAB}
	- \ii [G,N] + A = B,
\end{equation}
with $[B,N] = 0$. {In this latter equation and from now on we will assume that $N \in \OOO_\kappa$ and $N=\sum_{x \in \Lambda} N_x$ with $N_x$ supported on $\{x\}$ only.}

\begin{lemma}\label{lem:Homological}
	Let $A \in \OOO_\kappa$ {be self-adjoint}. A solution to \eqref{eq:homEqGenericAB} is given by {the self-adjoint operators} $B=\sum_{S \in \mathcal{P}_c(\Lambda)} B_S$ and $G=\sum_{S \in \mathcal{P}_c(\Lambda)} G_S$ with
	\begin{align}
		\label{eq:BSHomo} B_S &= \frac{1}{2\pi} \int_0^{2\pi} e^{-\ii N \vartheta} A_S \, e^{\ii N \vartheta} \, d\vartheta, \\
		\label{eq:GSHomo} G_S &= \frac{1}{2\pi} \int_0^{2\pi} \vartheta \, e^{-\ii N \vartheta} (A_S - B_S) \, e^{\ii N \vartheta} \, d\vartheta,
	\end{align}
	for each connected set $S \in \mathcal{P}_c(\Lambda)$, so that
	\begin{equation}\label{eq:ETR-303}
		[B,N] = 0, \qquad \Vert B \Vert_\kappa \leq \Vert A \Vert_\kappa, \qquad \Vert G \Vert_\kappa \leq 2\pi \Vert A \Vert_\kappa.
	\end{equation}
\end{lemma}

\begin{proof} {We first note that, if $A_S$ is self-adjoint, so are $B_S$ and $G_S$.} 	Let us {now} show that $B_S$ is such that $[B_S,N]=0$ for all $S \in \mathcal{P}_c(\Lambda)$. Indeed,
	\[
		\begin{split}
			[B_S,N]&=\left[\frac{1}{2\pi} \int_0^{2\pi} e^{-\ii N \vartheta} A_S e^{\ii N \vartheta} \, d \vartheta,N \right]=\frac{1}{2\pi} \int_0^{2\pi} e^{-\ii N \vartheta} [A_S,N] e^{\ii N \vartheta} \, d \vartheta \\
			&=-\frac{\ii}{2 \pi} \int_0^{2 \pi} \frac{d}{d\vartheta}(e^{-\ii N \vartheta} A_S e^{\ii N \vartheta} ) \, d \vartheta=0 \, ,
		\end{split}
	\]
	because $N$ has periodic flow. 
	Therefore $[B_S,N]=0$ for all $S$ and thus $[B,N]=\sum_{S \in \mathcal{P}_c(\Lambda)} [B_S,N]=0$. {The two norm bounds in \eqref{eq:ETR-303} follow from} \eqref{eq:GSHomo} and \eqref{eq:BSHomo}:
	\[
		\Vert B_S \Vert_{\mathrm{op}} \leq \frac{1}{2 \pi} \int_0^{2 \pi} \Vert A_S \Vert_{\mathrm{op}} \, d \vartheta = \Vert A_S \Vert_{\mathrm{op}} \, , \qquad \Vert G_S \Vert_{\mathrm{op}} \leq \frac{1}{2 \pi} \int_0^{2 \pi} \vartheta \Vert A_S-B_S \Vert_{\mathrm{op}} \, d\vartheta \leq 2 \pi \Vert A_S \Vert_{\mathrm{op}} \, .
	\]
	It remains to show that $G_S$ and $B_S$ solve the homological equation {\eqref{eq:homEqGenericAB}}. Indeed,
	\begin{equation*}
		\begin{split}
			-\ii [G_S,N]&=-\frac{\ii}{2 \pi} \int_0^{2\pi} \vartheta e^{-\ii N \vartheta} [A_S-B_S,N] e^{\ii N \vartheta} \, d \vartheta =-\frac{1}{2 \pi} \int_0^{2 \pi} \vartheta \frac{d}{d\vartheta}(e^{-\ii N \vartheta} (A_S-B_S) e^{\ii N \vartheta}) \, d \vartheta \\
			&=-\frac{1}{2\pi} \left( \vartheta e^{-\ii N \vartheta} (A_S-B_S) e^{\ii N \vartheta} \right|_{\vartheta=0}^{\vartheta=2 \pi}+\frac{1}{2 \pi} \int_0^{2 \pi} e^{-\ii N \vartheta} (A_S-B_S) e^{\ii N \vartheta} \, d \vartheta \\
			&=B_S-A_S \, ,
		\end{split}
	\end{equation*}
	where in the last step we used the definition of $B_S$ in \eqref{eq:BSHomo}. 
\end{proof}

\subsection{Iterative result and consequences}

\begin{lemma}\label{lem:Iterative}
Let $0 < \varepsilon < \varepsilon_0$, and define
\begin{equation}\label{eq:ETR-Polifemo}
n_* := \left\lfloor \frac{\varepsilon_0}{\varepsilon} \right\rfloor, \quad
\delta := \frac{{\kappa }}{4 n_*}, \quad
\kappa_n := {\kappa  - \delta n}.
\end{equation}
 Then, for any $n=0,\dots,n_*-1$ there exists {a self-adjoint operator} $G^{(n)} {\in \OOO_{\kappa_n}}$. {Moreover, for any $n=0,\dots,n_*$ there exists {a self-adjoint operator} $H^{(n)} \in \OOO_{\kappa_n}$} such that $H^{(n+1)}=e^{-\ii G^{(n)}} H^{(n)} e^{\ii G^{(n)}}$ and, for any $n=0,\dots,n_*$, $H^{(n)}$ admits the decomposition
	\begin{equation}\label{eq:ETR-252}
		H^{(n)}=N+Z^{(n)}+P^{(n)}
	\end{equation}
	{for $Z^{(n)}$ and $P^{(n)}$ self-adjoint} with the following properties:
	\begin{itemize}
		\item[(i)] $Z^{(0)}=0$, $P^{(0)}=\varepsilon P$ and for each $n=1,\dots,n_*$
		\begin{equation}
		\label{eq:IterativoPZ}
			\Vert Z^{(n)} \Vert_{\kappa_{n-1}} \leq \varepsilon \sum_{n'=0}^{n-1} e^{-n'} \Vert P \Vert_{\kappa} \, , \qquad [Z^{(n)},N]=0 \, , \qquad \Vert P^{(n)}\Vert_{\kappa_n} \leq \varepsilon e^{-n} \Vert P \Vert_{\kappa} \, ;
		\end{equation}
		\item[(ii)] for each $n=0,\dots, n_*-1$, $G^{(n)}$ is {self-adjoint,} small in size and its flow generates a $\varepsilon$-close-to-identity unitary transformation, that is 
		\begin{equation}\label{eq:IterativoG}
			\Vert G^{(n)} \Vert_{\kappa_n} \leq 2 \pi \varepsilon e^{-n} \Vert P \Vert_{\kappa} \, , \qquad \Vert e^{-\ii G^{(n)}} A e^{\ii G^{(n)}}-A \Vert_{\rho} \leq  K_{\rho} \varepsilon e^{-n} \Vert A \Vert_{\frac{3 \kappa}{4}} \, , \quad \forall A \in \OOO_{\frac{3 \kappa}{4}} \, , 
		\end{equation}
	where $K_{\rho}=8 \pi C_{\frac{1}{2}} \frac{e^{-\rho}}{\kappa} \Vert P \Vert_{\kappa}$ and {for any} $\rho \leq \frac{\kappa}{2}$. 	\end{itemize}
\end{lemma}

\begin{proof}
	The proof proceeds by induction and we start proving Items \emph{(i)} and \emph{(ii)} for the case $n=1$. By construction
	\begin{equation}
		\begin{split}
			H^{(1)}&=e^{-\ii G^{(0)}}H^{(0)} e^{\ii G^{(0)}} \\
			&=N+\left(-\ii [G^{(0)},N]+P^{(0)}\right) \\
			&\quad+\left(e^{-\ii G^{(0)}} N e^{\ii G^{(0)}}-N+\ii [G^{(0)},N]+e^{-\ii G^{(0)}} P^{(0)}e^{\ii G^{(0)}}-P^{(0)} \right) \, .
		\end{split}
	\end{equation}
	We choose $G^{(0)}$ and $\widetilde{Z}^{(0)}$ as solutions of the homological equation
	\begin{equation}\label{eq:IterativeHomProof}
		- \ii [G^{(0)},N]+P^{(0)}=\widetilde{Z}^{(0)} \, .
	\end{equation}
	By Lemma \ref{lem:Homological} (solution to the homological equation), we have that 
	\begin{equation}\label{eq:Robespierre}
		[\widetilde{Z}^{(0)},N]=0 \, , \qquad \Vert \widetilde{Z}^{(0)} \Vert_{\kappa} \leq \Vert P^{(0)} \Vert_{\kappa} = \varepsilon \Vert P \Vert_{\kappa} \, , \qquad \Vert G^{(0)} \Vert_{\kappa} \leq 2 \pi \varepsilon \Vert P \Vert_{\kappa} \, .
	\end{equation}
	 {Setting} $Z^{(1)}:=\widetilde{Z}^{(0)}$, {proves the first bound and the commutator relation in} \eqref{eq:IterativoPZ}, {as well as the first bound in}  \eqref{eq:IterativoG}. We define
	\begin{equation}
		P^{(1)}:=e^{-\ii G^{(0)}} N e^{\ii G^{(0)}}-N+\ii [G^{(0)},N]+e^{-\ii G^{(0)}} P^{(0)}e^{\ii G^{(0)}}-P^{(0)} \, .
	\end{equation}
	We have that
	\begin{equation}\label{eq:E-626}
		\frac{4 e^{-\kappa_1} \Vert G^{(0)} \Vert_{\kappa}}{ \delta} \leq \frac{16 e^{-\kappa_1} n_* \Vert G^{(0)} \Vert_\kappa}{{\kappa}} \overset{\text{\eqref{eq:Robespierre}}}{\leq}  \frac{32 \pi { \varepsilon \Vert P \Vert_{\kappa} \varepsilon_0}}{{\kappa}\varepsilon} < \frac{1}{2}
	\end{equation}
	where in the first {inequality} we {used} the definition of $n_*$ {in \eqref{eq:ETR-Polifemo}}, in the second {inequality} we used $e^{-\kappa_1} < 1$, $\Vert G^{(0)} \Vert_{\kappa} \leq 2 \pi \varepsilon \Vert P \Vert_{\kappa}$ and the explicit expression of $\delta$ and {the final inequality follows from the definition of} $\varepsilon_0$ (see \eqref{eq:EpsZero}). To control the size of $P^{(1)}$ we can apply Lemma \ref{lem:Commutatori} with $\eta=\frac{1}{2}$ to obtain
	\begin{equation}
		\begin{split}
			\Vert P^{(1)} \Vert_{\kappa_1} &\leq \Vert e^{-\ii G^{(0)}} N e^{\ii G^{(0)}}-N+\ii [G^{(0)},N] \Vert_{\kappa_1} +\Vert e^{-\ii G^{(0)}} P^{(0)}e^{\ii G^{(0)}}-P^{(0)} \Vert_{\kappa_1} \\
			&\hspace{-0.4cm}\overset{\text{\eqref{eq:CommEstimate}, \eqref{eq:FirstOrderEstimate}}}{\leq} \frac{C_{\frac{1}{2}} e^{-\kappa_1}}{\delta} \Vert G^{(0)} \Vert_{\kappa} \Vert [G^{(0)},N] \Vert_{\kappa}+\frac{C_{\frac{1}{2}}e^{-\kappa_1}}{\delta} \Vert G^{(0)} \Vert_{\kappa} \Vert P^{(0)} \Vert_{\kappa} \, .
		\end{split}
	\end{equation}
	By \eqref{eq:IterativeHomProof} we have $[G^{(0)},N]=\ii (\tilde{Z}^{(0)}-P^{(0)})$, and thus $[G^{(0)},N] \in \OOO_\kappa$ and $\Vert [G^{(0)},N] \Vert_{\kappa} \leq \Vert P^{(0)}- \widetilde{Z}^{(0)} \Vert_{\kappa}\leq 2 \Vert P^{(0)} \Vert_{\kappa}$. {We therefore obtain}
	\begin{equation}
		\begin{split}
			\Vert P^{(1)} \Vert_{\kappa_1} &\leq \frac{3 C_{\frac{1}{2}} e^{-\kappa_1}}{\delta} \Vert G^{(0)} \Vert_{\kappa} \Vert P^{(0)} \Vert_{\kappa} \overset{\delta=\frac{{\kappa}}{4n_*}{,e^{-\kappa_1}\leq 1}}{\leq} \frac{12 C_{\frac{1}{2}} n_*\Vert G^{(0)} \Vert_{\kappa} \Vert P^{(0)} \Vert_{\kappa}}{\kappa} \\
			&\hspace{-0.6cm}\overset{ \text{\eqref{eq:Robespierre}},n_*\leq \frac{\varepsilon_0}{\varepsilon}}{\leq} \frac{24 \pi C_{\frac{1}{2}} \varepsilon_0 \varepsilon \Vert P \Vert_{\kappa}^2}{\kappa} \leq e^{-1} \varepsilon \Vert P \Vert_\kappa
		\end{split}
	\end{equation}
	where in the second inequality we estimated $e^{-\kappa_1} \leq 1$ and we wrote $\delta$ in terms of $n_*$ {and $\kappa$}; in the third inequality we used the estimates for $G^{(0)}$ and $P^{(0)}$ and in the last inequality we used the bound on $\varepsilon_0$ (see eq. \eqref{eq:EpsZero}). This completes the proof of Item \emph{(i)} for $n=1$. Concerning Item \emph{(ii)}, {we first choose $\sigma=\frac{\kappa}{4}$ and we note that for any $\rho\leq \frac{\kappa}{2}$ and any $n=0,\dots,n_*-1$
\begin{equation}\label{eq:LaPoliziaSiIncazza}
	\frac{4 e^{-\rho} \Vert G^{(n)} \Vert_{\rho+\sigma}}{\sigma} \overset{\text{\eqref{eq:IterativoG}}}{\leq} \frac{32 e^{-\rho}  \pi \varepsilon e^{-n} \Vert P \Vert_{\kappa}}{\kappa} \overset{\text{\eqref{eq:EpsZero}}}{<}\frac{1}{2} \, .
\end{equation}
}We can thus apply Lemma \ref{lem:Commutatori}, which yields
	\begin{equation}
		\Vert e^{-\ii G^{(0)}} A e^{\ii G^{(0)}} - A \Vert_{\rho} \leq \frac{C_{\frac{1}{2}} e^{-\rho}}{\sigma} \Vert G^{(0)} \Vert_{\rho+\sigma} \Vert A \Vert_{\rho+\sigma} \leq 8 \pi C_{\frac{1}{2}} \frac{e^{-\rho}}{\kappa} \varepsilon \Vert P \Vert_{\kappa} \Vert A \Vert_{\frac{3 \kappa}{4}} \, .
	\end{equation}
	{Here} we chose $\sigma=\frac{\kappa}{4}$
and in the last step we used $\Vert G^{(0)}\Vert_{\rho+\sigma} \leq \Vert G^{(0)}\Vert_{\kappa}${. This proves Item \emph{(ii)} for the case $n=1$.}

For the inductive step, {assume that Items \emph{(i)} and \emph{(ii)} hold for some $1 \leq n<n_*$. We prove them for $n+1$.}

By construction
\begin{equation}
	\begin{split}
		H^{(n+1)}&=e^{-\ii G^{(n)}} H^{(n)} e^{\ii G^{(n)}} \\
		&=N+Z^{(n)}+ \left(-\ii [G^{(n)},N]+P^{(n)} \right)+\Big(e^{-\ii G^{(n)}}Ne^{\ii G^{(n)}}-N+\ii [G^{(n)},N] + \\
		&\quad+e^{-\ii G^{(n)}}P^{(n)} e^{\ii G^{(n)}}-P^{(n)}+e^{-\ii G^{(n)}} Z^{(n)} e^{\ii G^{(n)}} -Z^{(n)} \Big) \, .
	\end{split}
\end{equation}
We choose $G^{(n)}$ and $\widetilde{Z}^{(n)}$ as solutions of the homological equation
\begin{equation}\label{eq:IterativeHomProof-E-428}
	-\ii [G^{(n)},N]+P^{(n)}=\widetilde{Z}^{(n)} \, .
\end{equation}
Lemma \ref{lem:Homological} and the inductive hypothesis {give} $[\widetilde{Z}^{(n)},N]=0$, $\Vert \widetilde{Z}^{(n)} \Vert_{\kappa_n} \leq \Vert P^{(n)} \Vert_{\kappa_n} \leq \varepsilon e^{-n} \Vert P \Vert_{\kappa}$ and $\Vert G^{(n)} \Vert_{\kappa_n} \leq 2 \pi \Vert P^{(n)} \Vert_{\kappa_n} \leq 2 \pi \varepsilon \, e^{-n} \Vert P \Vert_{\kappa}$. We now define $Z^{(n+1)}:=Z^{(n)}+\widetilde{Z}^{(n)}$ and using also the inductive hypothesis we see that
\begin{equation}
	\Vert Z^{(n+1)} \Vert_{\kappa_n} \leq \Vert Z^{(n)} \Vert_{\kappa_n}+\Vert \widetilde{Z}^{(n)} \Vert_{\kappa_n} \leq \varepsilon \sum_{n'=0}^{n-1} e^{-n'} \Vert P \Vert_\kappa + \varepsilon e^{-n} \Vert P \Vert_\kappa = \varepsilon \sum_{n'=0}^{n} e^{-n'} \Vert P \Vert_\kappa \, ,
\end{equation}
which proves the first estimate in \eqref{eq:IterativoPZ}. {The commutation relation follows from} $[Z^{(n+1)},N]=[Z^{(n)},N]+[\widetilde{Z}^{(n)},N]=0$. We define
\begin{equation}
	\begin{split}
		P^{(n+1)}&:=e^{-\ii G^{(n)}}Ne^{\ii G^{(n)}}-N+\ii [G^{(n)},N] + \\
		&\quad+e^{-\ii G^{(n)}}P^{(n)} e^{\ii G^{(n)}}-P^{(n)}+e^{-\ii G^{(n)}} Z^{(n)} e^{\ii G^{(n)}} -Z^{(n)} \, .
	\end{split}
\end{equation}
Since
\begin{equation}\label{eq:E-428-aerodinamica}
	\frac{4 e^{-\kappa_{n+1}} \Vert G^{(n)} \Vert_{\kappa_n}}{\delta} \leq \frac{32 \pi e^{-n} \varepsilon \varepsilon_0 \Vert P \Vert_{\kappa}}{{\kappa}\varepsilon} < \frac{1}{2} \, ,
\end{equation}
where in the first step we used $e^{-\kappa_{n+1}} < 1$, $\Vert G^{(n)} \Vert_{\kappa_n} \leq 2 \pi e^{-n} \varepsilon \Vert P \Vert_\kappa$ and the explicit expression of $\delta$; in the last step {we used} the bound on $\varepsilon_0$ (see \eqref{eq:EpsZero}). We can thus apply Lemma \ref{lem:Commutatori} with $\eta=\frac{1}{2}$ and we get
\begin{equation}
	\begin{split}
		\Vert P^{(n+1)} \Vert_{\kappa_{n+1}} &\leq \Vert e^{-\ii G^{(n)}}Ne^{\ii G^{(n)}}-N+\ii [G^{(n)},N]\Vert_{\kappa_{n+1}} + \\
		&\quad+\Vert e^{-\ii G^{(n)}}P^{(n)} e^{\ii G^{(n)}}-P^{(n)}\Vert_{\kappa_{n+1}}+\Vert e^{-\ii G^{(n)}} Z^{(n)} e^{\ii G^{(n)}} -Z^{(n)}\Vert_{\kappa_{n+1}} \\
		&\hspace{-0.4cm}\overset{\text{\eqref{eq:CommEstimate}, \eqref{eq:FirstOrderEstimate}}}{\leq} \frac{C_{\frac{1}{2}} e^{-\kappa_{n+1}} \Vert G^{(n)} \Vert_{\kappa_n}}{\delta}( \Vert [G^{(n)},N] \Vert_{\kappa_{n}}+\Vert P^{(n)} \Vert_{\kappa_n} +\Vert Z^{(n)} \Vert_{\kappa_n}) \, .
	\end{split}
\end{equation}
Using now  \eqref{eq:IterativeHomProof-E-428}, one has $[G^{(n)},N]=\ii(P^{(n)}-\widetilde{Z}^{(n)})$, and therefore $\Vert [G^{(n)},N] \Vert_{\kappa_n} \leq \Vert P^{(n)}-\widetilde{Z}^{(n)} \Vert_{\kappa_n} \leq 2 \Vert P^{(n)} \Vert_{\kappa_n}$. We thus have
\begin{equation}
\Vert P^{(n+1)} \Vert_{\kappa_{n+1}} \leq \frac{C_{\frac{1}{2}} e^{-\kappa_{n+1}} \Vert G^{(n)} \Vert_{\kappa_n}}{\delta}(3 \Vert P^{(n)} \Vert_{\kappa_n}+\Vert Z^{(n)} \Vert_{\kappa_n} ) \, .
\end{equation}
{Substituting the bound for $\Vert G^{(n)} \Vert_{\kappa_{n_*}}$, the induction bounds for $P^{(n)}$ and $Z^{(n)}$, and the definition of $n_*$, we get}
\begin{equation}
	\begin{split}
		\Vert P^{(n+1)} \Vert_{\kappa_{n+1}} &\leq \frac{C_{\frac{1}{2}} 8 \pi n_* \varepsilon e^{-n} \Vert P \Vert_\kappa}{{\kappa}} \left(3 \varepsilon e^{-n} \Vert P \Vert_{\kappa}+\varepsilon \sum_{n'=0}^{n-1} e^{-n'} \Vert P \Vert_\kappa\right) \\&\leq \frac{C_{\frac{1}{2}} 8 \pi \varepsilon_0 e^{-n} \Vert P \Vert_\kappa}{{\kappa}} \left(3 \varepsilon e^{-n} \Vert P \Vert_{\kappa}+\varepsilon \sum_{n'=0}^{n-1} e^{-n'} \Vert P \Vert_\kappa\right) \\
		&\leq \frac{8 C_{\frac{1}{2}} \pi \varepsilon_0 e^{-n} \Vert P \Vert_\kappa^2 \varepsilon}{{\kappa}} \left(3 e^{-1}+\frac{e}{e-1} \right) \\
		&\leq \frac{8 C_{\frac{1}{2}} \pi (3(e-1)+e^2) \Vert P \Vert_\kappa \varepsilon_0}{e(e-1) {\kappa}} \varepsilon e^{-n} \Vert P \Vert_\kappa \leq  \varepsilon e^{-(n+1)} \Vert P \Vert_\kappa \, ,
	\end{split}
\end{equation}
where in the last step we used the bound on $\varepsilon_0$ (see eq. \eqref{eq:EpsZero}).
This is the last estimate in \eqref{eq:IterativoPZ} for $n+1$ and completes the {inductive} proof of Item \emph{(i)}. {Since \eqref{eq:LaPoliziaSiIncazza} holds,} Lemma \ref{lem:Commutatori} {gives}
\begin{equation}
	\Vert e^{-\ii G^{(n)}} A e^{\ii G^{(n)}}-A \Vert_{\rho} \leq \frac{C_{\frac{1}{2}} e^{-\rho} \Vert G^{(n)} \Vert_{\kappa_n}}{\sigma} \Vert A \Vert_{\rho+{\sigma}} \leq 8 \pi C_{\frac{1}{2}} \frac{e^{-\rho}}{\kappa} \varepsilon e^{-n} \Vert P \Vert_\kappa  \Vert A \Vert_{\frac{3 \kappa}{4}} \, ,
\end{equation}
{where $\sigma=\frac{\kappa}{4}$. This last step} completes the proof.
\end{proof}

\begin{lemma}
\label{lem:UnitaryConjugation}
Let $n_*$ be as in Lemma \ref{lem:Iterative}, define $Y=e^{-\ii G^{(n_*-1)}} e^{-\ii G^{(n_*-2)}} \cdots e^{-\ii G^{(0)}}$. 
\begin{itemize}
	\item[(i)] For any $A \in \OOO_{\kappa}$, $YAY^* \in \OOO_{\frac{3 \kappa}{4}}$ and $\Vert Y A Y^* \Vert_{\frac{3 \kappa}{4}} \leq 4 \Vert A \Vert_{\kappa}$;
	\item[(ii)] for any $A \in \OOO_{\frac{3 \kappa}{4}}$, $Y^*AY \in \OOO_{\frac{\kappa}{2}}$ {and} $\Vert Y^* A Y \Vert_{\frac{\kappa}{2}} \leq 4 \Vert A \Vert_{\frac{3 \kappa}{4}}$;
	\item[(iii)] for any $A \in \OOO_{\frac{3 \kappa}{4}}$, $\Vert Y^* A Y - A \Vert_{\frac{\kappa}{4}} \leq 32 \pi C_{\frac{1}{2}} \frac{e}{e-1}  \frac{e^{-\frac{\kappa}{4}}}{\kappa} \Vert P \Vert_{\kappa} \varepsilon \Vert A \Vert_{\frac{3 \kappa}{4}}$.
	\item[(iv)] For any local observable $O$, $\Vert Y^* O Y -O\Vert_{\mathrm{op}} \leq 4 \pi \varepsilon \frac{e}{e-1} |S| \Vert P \Vert_{\kappa} \Vert O \Vert_{\mathrm{op}}$. 
\end{itemize}
\end{lemma}
\begin{proof}
	We first set $Q=\frac{32 \pi e^{-\frac{ \kappa}{2}}}{\kappa} \Vert P \Vert_{\kappa} \varepsilon_0$ and we notice that from the definition of $\varepsilon_0$ in Eq.~\eqref{eq:EpsZero} one has $Q<\frac{1}{2}$. 	To prove Item \emph{(i)} we start by noting that, estimate \eqref{eq:IterativoG} {gives}
	\begin{equation}\label{eq:Rudolfiana}
		\frac{4 e^{-\kappa_{n+1}} \Vert G^{(n)} \Vert_{\kappa_n}}{\delta} \overset{\text{\eqref{eq:IterativoG}, \eqref{eq:ETR-Polifemo}}}{\leq} \frac{32 \pi e^{-\frac{ \kappa}{2}}}	{\kappa} \Vert P \Vert_{\kappa} \varepsilon_0 e^{-n}=Q e^{-n} \, ,
	\end{equation}
	{where $\delta$ is the quantity defined in \eqref{eq:ETR-Polifemo}. 
	Equation \eqref{eq:Conjugation} therefore implies,} for any $A \in \OOO_{\kappa_{n}}$,
	\begin{equation}\label{eq:ETR-400}
		\Vert e^{-\ii G^{(n)}} A e^{\ii G^{(n)}} \Vert_{\kappa_{n+1}} \overset{\text{\eqref{eq:Conjugation}}}{\leq} \frac{1}{1-Q e^{-n}} \Vert A \Vert_{\kappa_n} \, .
	\end{equation}
	We thus argue iteratively: we define $Y_n:=e^{-\ii G^{(n-1)}} \cdots e^{-\ii G^{(0)}}$ and $Y_0:=\mathbbm{1}$ so that  $Y_{n_*}=Y$ and $Y_{n}=e^{-\ii G^{(n-1)}}Y_{n-1}$ for all $n=1,\dots,n_*$. We now show iteratively that 
	\begin{equation}\label{eq:Cisalpino}
		\Vert Y_n A Y_n^* \Vert_{\kappa_n} \leq \left(\prod_{j=0}^{n-1} \frac{1}{1-Q e^{-j}} \right) \Vert A \Vert_{\kappa} \, .
	\end{equation}
For $n=1$, this {result follows by recalling that $Y_1=e^{-\ii G^{(0)}}$ and} from \eqref{eq:ETR-400} with $G^{(n)}=G^{(0)}$. Assuming now that \eqref{eq:Cisalpino} holds for $n-1$, let us prove it for $n$. Indeed,
\begin{equation}
	\begin{split}
		\Vert Y_n A Y_n^* \Vert_{\kappa_n}&=\Vert e^{-\ii G^{(n-1)}} Y_{n-1} A Y_{n-1}^* e^{\ii G^{(n-1)}} \Vert_{\kappa_n} \overset{\text{\eqref{eq:ETR-400}}}{\leq} \frac{1}{1-Q e^{-(n-1)}} \Vert Y_{n-1} A Y_{n-1}^* \Vert_{\kappa_{n-1}} \\
		&\leq \frac{1}{1-Q e^{-(n-1)}}  \left( \prod_{j=0}^{n-2} \frac{1}{1-Qe^{-j}} \right) \Vert A \Vert_{\kappa}
	\end{split}
\end{equation}
which completes the iterative proof and shows that $\Vert Y A Y^*\Vert_{\frac{3 \kappa}{4}} \leq \left(\prod_{j=0}^{n_*-1} \frac{1}{1-Q e^{-j}} \right) \Vert A \Vert_{\kappa}$. {It remains to bound} $\prod_{j=0}^{n_*-1} \frac{1}{1-Q e^{-j}}$. {Taking logarithms on both sides, we find}
\[
	\ln\left(\prod_{j=0}^{n_*-1} \frac{1}{1-Q e^{-j}}\right) = -\sum_{j=0}^{n_*-1} \ln(1-Q e^{-j}) \, .
\]
{We use the convexity of the function $x \mapsto - \ln(1-Qx)$ on $[0,1]$:} $-\ln(1-Q x) \leq - \ln(1-Q) x$. Thus,
\[
	\ln\left(\prod_{j=0}^{n_*-1} \frac{1}{1-Q e^{-j}}\right) = -\sum_{j=0}^{n_*-1} \ln(1-Q e^{-j}) \leq (-\ln(1-Q)) \sum_{j=0}^{n_*-1} e^{-j} \leq (- \ln (1-Q)) \frac{e}{e-1} \, .
\]
We thus proved 
\begin{equation}\label{eq:ProductLogarithmPadovano}
	\prod_{j=0}^{n_*-1} \frac{1}{1-Q e^{-j}} \leq \left(\frac{1}{1-Q}\right)^{\frac{e}{e-1}} < 4
\end{equation}
where we used that $Q\leq \frac{1}{2}$ and $\frac{e}{e-1}<2$. This completes the proof of Item \emph{(i)}.

Concerning Item \emph{(ii)}, we use the same strategy but we need to start from $\kappa_{2n_*}$ because $ G^{(n_*-1)} $ is bounded only in $\kappa_{n_*}$ norm. We define $W_{n}:= e^{- \ii G^{(n_*-1)}} \cdots e^{- \ii G^{(n_*-n)}}$ so that $Y=W_{n_*}$ and $W_{n+1}=W_n e^{-\ii G^{(n_*-n-1)}}$. We prove iteratively that
\begin{equation}\label{eq:ETR-475}
	\Vert W_n^* A W_n \Vert_{\kappa_{n_{*}+n}} \leq \left(\prod_{j=n_*-1}^{n_*-n} \frac{1}{1-Qe^{-j}} \right) \Vert A \Vert_{\frac{3 \kappa}{4}} \, {.}
\end{equation}
For the case $n=1$, {we first estimate
\begin{equation}\label{eq:Aln663}
\frac{4 e^{-\kappa_{n_*+1}}\Vert G^{(n_*-1)}\Vert_{\kappa_{n_*}}}{\delta} \leq \frac{4 e^{-\frac{\kappa}{2}} \Vert G^{(n_*-1)} \Vert_{\kappa_{n_*-1}}}{\delta}\overset{\text{\eqref{eq:Rudolfiana}}}{\leq}Q e^{-(n_*-1)} < \frac{1}{2}
\end{equation}
and then}
\begin{equation*}
	\Vert W_1^* A W_1\Vert_{\kappa_{n_*+1}} = \Vert e^{\ii G^{(n_*-1)}} A e^{-\ii G^{(n_*-1)}} \Vert_{\kappa_{n_*+1}} \overset{\text{\eqref{eq:Conjugation}, \eqref{eq:Aln663}}}{\leq} \frac{1}{1-Q e^{-(n_*-1)}} \Vert A \Vert_{\frac{3 \kappa}{4}} \, .
\end{equation*}
 Now suppose that the iterative step holds up to $n$ and we prove it holds also for $n+1 {\leq} n_*$. {To perform the iteration, we first note that
 \[
 	\frac{4 e^{-\kappa_{n_*+n+1}}\Vert G^{(n_*-n-1)} \Vert_{\kappa_{n_*+n}}}{\delta} \overset{\text{\eqref{eq:Rudolfiana}}}{\leq} Q e^{-m} < \frac{1}{2}
 \]
 and therefore we can apply Lemma \ref{lem:Commutatori} to have}
\begin{equation*}
	\begin{split}
		\Vert W_{n+1}^* A W_{n+1} \Vert_{\kappa_{n_*+n+1}} &= \Vert e^{\ii G^{(n_*-n-1)}} W_n^* A W_n e^{-\ii G^{(n_*-n-1)}} \Vert_{\kappa_{n_*+n+1}} \\
		& \leq \frac{1}{1-Q e^{-(n_*-n-1)}} \Vert W_n^* A W_n \Vert_{\kappa_{n_*+n}} \\
		&\leq \frac{1}{1-Q e^{-(n_*-n-1)}}  \left(\prod_{j={n_*-1}}^{n_*-n} \frac{1}{1-Qe^{-j}}\right) \Vert A \Vert_{\frac{3 \kappa}{4}} \\
		&= \left( \prod_{j=n_*-1}^{n_*-(n+1)} \frac{1}{1-Q e^{-j}} \right) \Vert A \Vert_{\frac{3\kappa}{4}}
	\end{split}
\end{equation*}	 
which is the iterative step and implies immediately Item \emph{(ii)} by the use of \eqref{eq:ProductLogarithmPadovano}.	 

Concerning Item \emph{(iii)}, we define $W_0:=\mathbbm{1}$. Then,
\begin{equation*}
	\begin{split}
		\Vert Y^* A Y - A \Vert_{\frac{\kappa}{4}} & = \Big\Vert \sum_{n'=0}^{n_*-1} \left( W_{n'+1}^* A W_{n'+1}-W^*_{n'} A W_{n'}\right) \Big\Vert_{\frac{\kappa}{4}} \\
		&\leq \sum_{n'=0}^{n_*-1} \Vert W_{n'+1}^* A W_{n'+1}-W_{n'}^* A W_{n'} \Vert_{\frac{\kappa}{4}} \\
		&= \sum_{n'=0}^{n_*-1} \Vert e^{\ii G^{(n_*-n'-1)}} W_{n'}^* A W_{n'} e^{-\ii G^{(n_*-n'-1)}}-W_{n'}^* A W_{n'} \Vert_{\frac{\kappa}{4}} \\
		&\hspace{-2.5pt}\overset{\text{\eqref{eq:FirstOrderEstimate}}}{\leq} \sum_{n'=0}^{n_*-1} C_{\frac{1}{2}} e^{-\frac{\kappa}{4}} \frac{4}{\kappa} \Vert G^{(n_*-n'-1)} \Vert_{\frac{\kappa}{2}} \Vert W^*_{n'} A W_{n'} \Vert_{\frac{\kappa}{2}} \\
		& \hspace{-26pt}\overset{\text{\eqref{eq:IterativoG}, \eqref{eq:ETR-475}, \eqref{eq:ProductLogarithmPadovano}}}{\leq} 8 \pi C_{\frac{1}{2}} \frac{e^{-\frac{\kappa}{4}}}{\kappa} \varepsilon \sum_{n'=0}^{n_*-1} e^{-(n_*-n'-1)} \Vert P \Vert_{\kappa} \left(\frac{1}{1-Q}\right)^{\frac{e}{e-1}} \Vert A \Vert_{\frac{3 \kappa}{4}} \\
		& \leq 8 \pi C_{\frac{1}{2}} \frac{e}{e-1} \left(\frac{1}{1-Q} \right)^{\frac{e}{e-1}} \frac{e^{-\frac{\kappa}{4}}}{\kappa} \varepsilon \Vert P \Vert_{\kappa} \Vert A \Vert_{\frac{3 \kappa}{4}} \, .
	\end{split}
\end{equation*}
{ The third step uses $W_{n'}^* A W_{n'} \in \OOO_{\frac{\kappa}{2}}$, established in \eqref{eq:ETR-475}; the conclusion
then follows from \eqref{eq:ProductLogarithmPadovano}.}

Concerning Item \emph{(iv)}, {we start by noting that $Y$ is unitary in operator norm and thus $\Vert O-Y^*OY\Vert_{\mathrm{op}}=\Vert YOY^*-O\Vert_{\mathrm{op}}$. We now estimate the latter one. Indeed, since $O=O_S$,} 
\[
	\begin{split}
		\Vert Y O Y^*-O \Vert_{\mathrm{op}} & = \Vert W_{n_*} O_S W_{n_*}^*-O \Vert_{\mathrm{op}} \\
		&=\Big\Vert \sum_{n'=0}^{n_*-1} (W_{n'+1} O_S W_{n'+1}^*-W_{n'} O_S W_{n'}^*) \Big\Vert_{\mathrm{op}} \\
		&\leq \sum_{n'=0}^{n_*-1} \Vert W_{n'+1} O_S W_{n'+1}^* - W_{n'} O_S W_{n'}^* \Vert_{\mathrm{op}} \\
		&\leq \sum_{n'=0}^{n_*-1} \Vert W_{n'} e^{-\ii G^{(n_*-n'-1)}} O_S e^{\ii G^{(n_*-n'-1)}} W_{n'}^* - W_{n'}O_S W_{n'}^* \Vert_{\mathrm{op}} \\
		&=\sum_{n'=0}^{n_*-1} \Vert e^{-\ii G^{(n_*-n'-1)}} O_S e^{\ii G^{(n_*-n'-1)}} - O \Vert_{\mathrm{op}}=(\star) \, .
	\end{split}
\]
It remains to estimate each term of the sum,
\[
	\begin{split}
		\Vert e^{-\ii G^{(n_*-n'-1)}} O_S e^{\ii G^{(n_*-n'-1)}} - O \Vert_{\mathrm{op}} & \leq \int_0^1 \Vert e^{-\ii G^{(n_*-n'-1)} \tau} [G^{(n_*-n'-1)},O_S] e^{\ii G^{(n_*-n'-1)} \tau} \Vert_{\mathrm{op}} \, d \tau \\
		&=\Vert [G^{(n_*-n'-1)},O_S] \Vert_{\mathrm{op}} \, .
	\end{split}
\]
To estimate this last term, we exploit locality:
\[
	\begin{split}
		\Vert [G^{(n_*-n'-1)},O_S] \Vert_{\mathrm{op}} & = \Big\Vert \sum_{\substack{S' \in \mathcal{P}_c(\Lambda) \\ S' \cap S \neq \varnothing}} [G^{(n_*-n'-1)}_{S'},O_S] \Big\Vert_{\mathrm{op}} \leq \sum_{\substack{S' \in \mathcal{P}_c(\Lambda) \\ S' \cap S \neq \varnothing}} \Vert {[} G_{S'}^{(n_*-n'-1)},O_S] \Vert_{\mathrm{op}} \\
		& \leq \sum_{x \in S} \sum_{\substack{S' \in \mathcal{P}_c(\Lambda) \\ x \in S'}} 2 \Vert G^{(n_*-n'-1)}_{S'} \Vert_{\mathrm{op}} \Vert O_S \Vert_{\mathrm{op}} \leq |S| \sup_{x \in \Lambda} \sum_{\substack{S' \in \mathcal{P}_c(\Lambda) \\ x \in S'}} 2 \Vert G^{(n_*-n'-1)}_{S'} \Vert_{\mathrm{op}} \Vert O_S \Vert_{\mathrm{op}} \\
		&\leq 2 |S|\,  \Vert G^{(n_*-n'-1)} \Vert_{0} \Vert O \Vert_{\mathrm{op}} \, .
	\end{split}
\]
We thus have
\[
	\begin{split}
		(\star) &\leq 2 |S| \sum_{n'=0}^{n_*-1} \Vert G^{(n_*-n'-1)} \Vert_0 \Vert O \Vert_{\mathrm{op}} \overset{\text{
		\eqref{eq:IterativoG}}}{\leq} 4 \pi \varepsilon |S|\sum_{n'=0}^{n_*-1} e^{-(n_*-n'-1)} \Vert P \Vert_{\kappa} \Vert O \Vert_{\mathrm{op}} \\
		& \leq 4 \pi \varepsilon \frac{e}{e-1} |S| \Vert P \Vert_{\kappa} \Vert O \Vert_{\mathrm{op}} \, .
	\end{split}
\]
\end{proof}

{We require the following consequence of the Lieb–Robinson bounds; see
Lemma 2.15 of \cite{Gallone-Langella-Ck} or Lemma 6.2 in \cite{Gallone-Langella-2024}).}

\begin{lemma}\label{lem:LibRobinson}
	Let $O$ be a local observable supported on $S_O \in \mathcal{P}_c(\Lambda)$, {$A,B$ self-adjoint,} $A \in \OOO_\rho$ and $B \in \OOO_{2 \rho}$ for some $\rho>0$. Then there exists a positive constant $C_{LR}=C_{LR}(|S_O|,d,\rho)$ such that
	\begin{equation}
		\int_0^t \Vert [A,e^{\ii s B}Oe^{-\ii s B}] \Vert_{\mathrm{op}}  \, ds \leq C_{LR} \langle \Vert B \Vert_{2 \rho} \rangle^d \langle t \rangle^{d+1} \Vert O \Vert_{\mathrm{op}} \Vert A \Vert_{\rho} \, {.}
	\end{equation}
\end{lemma}

{\begin{lemma}Let $A \in \OOO_0$. Then
\begin{equation}\label{eq:E447}
	\Vert A \Vert_{\mathrm{op}} \leq |\Lambda| \Vert A \Vert_0 \, .
\end{equation}
\end{lemma}
\begin{proof}
	Using that $A=\sum_{S \in \mathcal{P}_c(\Lambda)} A_S$, we have
	\[
		\begin{split}
			\Vert A \Vert_{\mathrm{op}} &= \Big\Vert \sum_{S \in \mathcal{P}_c(\Lambda)} A_S \Big\Vert_{\mathrm{op}} \leq \sum_{S \in \mathcal{P}_c(\Lambda)} \Vert A_S \Vert_{\mathrm{op}} \leq \sum_{x \in \Lambda} \sum_{\substack{S \in \mathcal{P}_c(\Lambda) \\ x \in S}} \Vert A_S \Vert_{\mathrm{op}} \\
			&\leq |\Lambda| \sup_{x \in \Lambda} \sum_{\substack{S \in \mathcal{P}_c(\Lambda) \\ x \in S}} \Vert A_S \Vert_{\mathrm{op}}= |\Lambda| \Vert A \Vert_0 \, .
		\end{split}
	\]
	In the last step we used the definition of $\Vert \cdot \Vert_0$ norm \eqref{eq:ZioPinoloSecondo}.
\end{proof}
}

{We use the unitary operator $Y$ of Lemma \ref{lem:UnitaryConjugation} to define the \emph{dressed operators} $\mathcal{N}$ and $\mathcal{Z}$ appearing in Theorem \ref{thm:Main} as
\begin{equation}\label{eq:ETR-251}
	\mathcal{N}=Y^*NY \, , \qquad \mathcal{Z}=Y^*Z^{(n_*)}Y
\end{equation}
where $Z^{(n_*)}$ is defined in \eqref{eq:ETR-252}.}

\begin{proof}[Proof of Theorem \ref{thm:Main}] {Item (ii) of Lemma \ref{lem:UnitaryConjugation} implies that $\mathcal{N}=Y^*NY \in \OOO_{\frac{\kappa}{2}}$.}  The same holds for $\mathcal{Z}=Y^*Z^{(n_*)} Y$. {Since $Y$ is unitary on $\mathcal{H}_\Lambda$, unitary equivalence  gives $\sigma(\mathcal{N})=\sigma(Y^* N Y)=\sigma(N)$. This proves} Item \emph{(i)}. 

Concerning Item \emph{(ii)}, the stronger statement $\Vert \mathcal{N} - N \Vert_{{\frac{\kappa}{4}}} \leq C \varepsilon \Vert N \Vert_{\frac{3 \kappa}{4}}$ follows from Item \emph{(iii)} of Lemma \ref{lem:UnitaryConjugation}. Then,
\begin{equation*}
	\Vert \mathcal{N}-N \Vert_{\mathrm{op}} \overset{\text{\eqref{eq:E447}}}{\leq} |\Lambda| \Vert \mathcal{N}-N\Vert_{\frac{\kappa}{4}} \leq |\Lambda| C \varepsilon \Vert N \Vert_{\frac{3 \kappa}{4}}
\end{equation*}
with $C=32 \pi C_{\frac{1}{2}} \frac{e}{e-1}  \frac{e^{-\frac{\kappa}{4}}}{\kappa} \Vert P \Vert_{\kappa}$, {which proves Item (ii)} from which the thesis follows.

Then,
\begin{equation}
	\begin{split}
		\Vert e^{\ii H t} \mathcal{N} e^{-\ii H t} - \mathcal{N} \Vert_{\mathrm{op}} &= \Vert  e^{\ii H t} Y^* N Y e^{-\ii H t}  - Y^*N Y \Vert_{\mathrm{op}} \\
		&=  \Vert Y e^{\ii H t} Y^* N Y e^{-\ii H t} Y^*-N \Vert_{\mathrm{op}} = (\star) \, .
	\end{split}
\end{equation}
{By} Lemma \ref{lem:Iterative},  $e^{-\ii G^{(n)}} H^{(n)} e^{\ii G^{(n)}}=H^{(n+1)}$, {and therefore}
\begin{equation}
	\begin{split}
		(\star)&= \Vert e^{\ii H^{(n_*)} t} N e^{-\ii H^{(n_*)}{t}}-N \Vert_{\mathrm{op}}=\left\Vert \int_0^t e^{\ii H^{(n_*)} \tau}[H^{(n_*)},N] e^{-\ii H^{(n_*)} \tau} \,\ud \tau \right\Vert_{\mathrm{op}} \\
		&\leq |t| \Vert [P^{(n_*)},N] \Vert_{\mathrm{op}} \overset{\text{\eqref{eq:E447}}}{\leq} |t| |\Lambda| \Vert [P^{(n_*)},N] \Vert_{0} \overset{\text{\eqref{eq:ALe601}}}{\leq} |t| |\Lambda|\frac{1}{e} \frac{16}{3 \kappa} \Vert P^{(n_*)} \Vert_{\frac{3\kappa}{4}} \Vert N \Vert_{\frac{3 \kappa}{4}} \, .
	\end{split}
\end{equation}
{The last inequality follows from Lemma \ref{lem:D-345} with $j=1$. Combining
Lemma \ref{lem:Iterative} with the definition of $n_*$ in \eqref{eq:ETR-Polifemo} gives}
\begin{equation}
	\Vert e^{\ii H t} \mathcal{N} e^{-\ii H t} - \mathcal{N} \Vert_{\mathrm{op}} \leq {\frac{16}{3\kappa} \Vert P \Vert_{\kappa} \Vert N \Vert_{\frac{3 \kappa}{4}} |\Lambda| \varepsilon e^{-\frac{\varepsilon_0}{\varepsilon}} |t|}
\end{equation}
which proves the {first estimate in} Item \emph{(iii)}. 	{The argument for $\mathcal{Z}$ is analogous}:
	\[
		\begin{split}
		\Vert e^{\ii H t} \mathcal{Z} e^{-\ii H t}- \mathcal{Z} \Vert_{\mathrm{op}}&=\Vert Y e^{\ii H t} Y^* Z^{(n_*)} Y e^{-\ii H t} Y^*-Z^{(n_*)} \Vert_{\mathrm{op}}=\Vert e^{\ii H^{(n_*)} t} Z^{(n_*)} e^{-\ii H^{(n_*)} t} - Z^{(n_*)} \Vert_{\mathrm{op}} \\
		&\leq \int_0^t \Vert [Z^{(n_*)},H^{(n_*)}] \Vert_{\mathrm{op}} \, ds = |t|\Vert[Z^{(n_*)},P^{(n_*)}] \Vert_{\mathrm{op}} \\
		&\hspace{-5pt}\overset{\text{\eqref{eq:E447}}}{\leq} |t| |\Lambda| \Vert [Z^{(n_*)},P^{(n_*)}] \Vert_{{0}} \overset{\text{\eqref{eq:ALe601}}}{\leq} |t| |\Lambda| \frac{16}{3 e \kappa} \Vert P^{(n_*)}\Vert_{\frac{3 \kappa}{4}} \Vert Z^{(n_*)} \Vert_{\frac{3 \kappa}{4}} \\ &
		\overset{\text{\eqref{eq:IterativoPZ}}}{\leq} |t| |\Lambda|\frac{16}{3 e \kappa} \frac{e^2}{e-1} \varepsilon^2 e^{-\frac{\varepsilon_0}{\varepsilon}} {\Vert P \Vert_{\kappa}^2}
		\end{split}
	\]
where we used Lemma \ref{lem:D-345} and this completes the proof of Item \emph{(iii)}.

Item \emph{(iv)} follows from Item \emph{(i)} of Lemma \ref{lem:Iterative}, indeed
\[
	[\mathcal{N},\mathcal{Z}]=Y^* [N,Z^{(n_*)}] Y=0 \, .
\]

For Item \emph{(v)} we use {unitary invariance of the operator norm} to write
\[
	\begin{split}
		\Vert e^{\ii H t} &O e^{-\ii H t} - e^{\ii (\mathcal{N}+\mathcal{Z})t} O e^{-\ii (\mathcal{N}+\mathcal{Z})t} \Vert_{\mathrm{op}} \\
		&{= \Vert e^{\ii H t}(O-Y^*OY+Y^*OY)e^{-\ii H t}-Y^* e^{\ii(N+Z^{(n_*)})t} Y O Y^* e^{-\ii (N+Z^{(n_*)}) t} Y \Vert_{\mathrm{op}}} \\
		&{= \Vert e^{\ii H t}(O-Y^*OY+Y^*OY)e^{-\ii H t}-Y^* e^{\ii(N+Z^{(n_*)})t} (Y O Y^*-O+O) e^{-\ii (N+Z^{(n_*)}) t} Y\Vert_{\mathrm{op}}} \\
		&\leq 2 \Vert O- Y^* O Y \Vert_{\mathrm{op}}+\Vert e^{\ii H t} Y^* O Y e^{-\ii H t}-Y^* e^{\ii(N+Z^{(n_*)})t} O e^{-\ii (N+Z^{(n_*)})t}Y \Vert_{\mathrm{op}} \\
		&\leq 2 \Vert O-Y^*OY\Vert_{\mathrm{op}}+\Vert Y e^{\ii H t} Y^* O Y e^{-\ii H t}Y^*- e^{\ii(N+Z^{(n_*)})t} O e^{-\ii (N+Z^{(n_*)})t} \Vert_{\mathrm{op}}\\
		&\leq 2 \Vert O - Y^*OY \Vert_{\mathrm{op}} + \Vert e^{\ii H^{(n_*)}t} O e^{-\ii H^{(n_*)} t}-e^{\ii (N+Z^{(n_*)})t} O e^{-\ii (N+Z^{(n_*)})t} \Vert_{\mathrm{op}} \, .
	\end{split}
\]
Concerning the first item, we use Item \emph{(iv)} of Lemma \ref{lem:UnitaryConjugation} to have
\[
	2 \Vert O - Y^* O Y \Vert_{\mathrm{op}} \leq 8 \pi \varepsilon \frac{e}{e-1} |S| \Vert P \Vert_{\kappa} \Vert O \Vert_{\mathrm{op}} \, .
\]
Concerning the second term, {Duhamel's Formula and Lemma \ref{lem:LibRobinson} give} 
\[
	\begin{split}
		\Vert e^{\ii H^{(n_*)} t} O e^{-\ii H^{(n_*)} t} - e^{\ii (N+Z^{(n_*)}) t} O &e^{-\ii (N+Z^{(n_*)}) t} \Vert_{\mathrm{op}} =\\
		& = \Vert e^{-\ii H^{(n_*)}t} e^{\ii (N+Z^{(n_*)})t} O e^{-\ii (N+Z^{(n_*)})t}e^{\ii H^{(n_*)} t}-O \Vert_{\mathrm{op}} \\
		&\leq \int_0^t \Vert [P^{(n_*)},e^{\ii(N+Z^{(n_*)}) \tau} O e^{-\ii (N+Z^{(n_*)}) \tau} ]\Vert_{\mathrm{op}} \, d \tau \\ &\leq C_{LR} \langle \Vert N \Vert_{\frac{\kappa}{2}}+ \Vert Z^{(n_*)} \Vert_{\frac{\kappa}{2}} \rangle^d \langle t \rangle^{d{+1}} \Vert P^{(n_*)} \Vert_{\frac{\kappa}{4}} \Vert O \Vert_{\mathrm{op}} \\
		&\hspace{-4pt}\overset{\text{\eqref{eq:IterativoPZ}}}{\leq} C_{LR} \langle \Vert N \Vert_{\kappa} + \Vert Z^{(n_*)} \Vert_{\kappa_{n_*-1}} \rangle^d \langle t \rangle^{d+1} \varepsilon \, e \, e^{-\frac{\varepsilon_0}{\varepsilon}} \Vert P \Vert_{\kappa} \Vert O \Vert_{\mathrm{op}} \, .
	\end{split}
\]
Using the inequality $(1+|t|)^{d+1} \leq 2^d (1+|t|^{d+1})$,  $\varepsilon<\varepsilon_0$ and the estimate \eqref{eq:IterativoPZ} for $\Vert Z^{(n_*)} \Vert_{\kappa_{n_*-1}}$ we get the thesis.

\end{proof}

\bigskip

\footnotesize

\noindent
\textbf{Acknowledgments.} The author warmly thanks B.~Langella for many stimulating discussions related to the topic of this paper.

\medskip

\noindent
\textbf{Conflicts of Interest.}
The author declares no conflicts of interest.

\medskip

\noindent
\textbf{Data Availability Statement.} 
Data sharing is not applicable to this article as no data sets were generated or analyzed during this study.

\bigskip

\bibliographystyle{siam}
\bibliography{DraftBib.bib}

\end{document}